\documentclass[11pt, letterpaper]{article}
\usepackage{graphicx} %
\usepackage{amsmath, amsthm, mathtools, amssymb}
\usepackage{xcolor}
\usepackage{hyperref}
\usepackage{fullpage}
\usepackage[
  style=alphabetic,
  maxbibnames=99,
  maxalphanames=4, 
  minalphanames=3,
  sortcites=false,   %
]{biblatex}
\usepackage{authblk}

\newtheorem{theorem}{Theorem}[section]

\newtheorem{proposition}[theorem]{Proposition}
\newtheorem{definition}[theorem]{Definition}
\newtheorem{claim}[theorem]{Claim}

\newtheorem{corollary}[theorem]{Corollary}

\def\cA{{\cal A}}
\def\cB{{\cal B}}

\def\cD{{\cal D}}

\def\cF{{\cal F}}

\def\cH{{\cal H}}

\def\cU{{\cal U}}

\def\cX{{\cal X}}
\def\cY{{\cal Y}}

\def\bbC{{\mathbb C}}
\def\bbE{{\mathbb E}}

\def\bbN{{\mathbb N}}

\newcommand{\E}{\bbE}

\newcommand{\bv}[1]{\boldsymbol{#1}}

\newcommand{\ceil}[1]{\left\lceil #1 \right\rceil}
\newcommand{\floor}[1]{\left\lfloor #1 \right\rfloor}

\newcommand{\hbd}{\mathsf{H}}

\def\binset{\{0,1\}}

\newcommand{\abs}[1]{\left\vert {#1} \right\vert}
\newcommand{\norm}[1]{\left\| {#1} \right\|}
\newcommand{\set}[1]{\left\{ {#1} \right\}}

\def\poly{{\rm poly}}

\def\negl{{\rm negl}}

\newcommand{\tr}{\mathop{\mathrm{Tr}}}
\newcommand{\td}[2]{\mathop{\mathrm{TD}} \left({#1}, {#2} \right)}
\newcommand{\trnorm}[1]{{\norm{{#1}}_1}}
\newcommand{\trnormdis}[2]{\norm{{#1} - {#2}}_1}

\newcommand{\cind}{{\ \overset{c}{\approx} \ }}
\newcommand{\sind}{{\ \overset{s}{\approx} \ }}

\newcommand{\tsimi}[2]{\textrm{Sim}_{#1}^t \left({#2} \right)}
\newcommand{\tsim}[1]{\textrm{Sim}^t \left({#1} \right)}

\newcommand{\secp}{\lambda}
\newcommand{\keylen}{k}
\newcommand{\outqubits}{m}
\newcommand{\garbqubits}{\mathsf{a}}
\newcommand{\designseedlen}{\kappa}
\newcommand{\s}{s}
\newcommand{\epsExt}{\varepsilon}
\newcommand{\nExt}{n'}
\newcommand{\kExt}{k'}

\newcommand{\tParam}[1]{\widetilde{#1}}
\newcommand{\ampParam}[1]{\widehat{#1}}

\newcommand{\dist}{\mathrm{dis}(n,t)}
\newcommand{\distproj}{{\Pi_{\mathrm{dis}}^{n, t}}}
\newcommand{\distprojreg}[1]{\Pi_{\mathrm{dis}, {#1}}^{n, t}}

\newcommand{\collprojreg}[1]{\Pi_{\mathrm{coll}, {#1}}^{n, t}}

\newcommand{\symparam}[2]{\mathrm{Sym}({#1},{#2})}
\newcommand{\symprojparam}[2]{\Pi_{\mathrm{Sym}}^{{#1},{#2}}}
\newcommand{\symrhoparam}[2]{\rho_{\mathrm{Sym}}^{{#1},{#2}}}

\newcommand{\symproj}{\symprojparam{\tParam{\outqubits}}{t}}
\newcommand{\symrho}{\symrhoparam{\tParam{\outqubits}}{t}}

\newcommand{\ket}[1]{|{#1}\rangle}
\newcommand{\bra}[1]{\langle{#1}|}
\newcommand{\braket}[2]{\langle{#1}|{#2}\rangle}
\newcommand{\ketbra}[2]{\ket{{#1}}\bra{{#2}}}

\newcommand{\vi}{\bv{i}}
\newcommand{\ki}{\ket{i}}
\newcommand{\kvi}{\ket{\vi}}
\newcommand{\kgi}{\ket{g_i}}
\newcommand{\kpx}[1]{\ket{\varphi_{{#1}}}}
\newcommand{\kgx}[1]{\ket{g_{{#1}}}}
\newcommand{\kpi}{\kpx{i}}
\newcommand{\glen}{\garbqubits}

\title{Generic Number-of-Copies Amplification for Pseudorandom States}
\author{Zvika Brakerski}
\author{Miri Zenilman}
\affil{Weizmann Institute of Science\\
\small\texttt{\{zvika.brakerski, miri.butel\}@weizmann.ac.il}}
\date{}

\begin{document}

\maketitle

\begin{abstract}
	We show that \emph{any} quantum pseudorandom state that is secure against single-copy distinguishers, i.e.\ a $1$-PRS, can be amplified to $t$-copy security, i.e.\ to a $t$-PRS, without additional assumptions, for any polynomial $t$ in the security parameter. Prior work (Ananth and Goldin, CRYPTO 2026) was only able to show this for a restricted class of $1$-PRS constructions, namely ones whose generators only use a small number of ancilla qubits.
	
	Technically, we show that by carefully accounting for the randomness that is used in the construction, and using quantum extractors, it is possible to eliminate an ancilla register of any length and obtain a meaningful $t$-PRS outcome.
\end{abstract}

\tableofcontents

\newpage

\section{Introduction}

Research in recent years uncovered that the complexity-theoretic landscape of quantum cryptography is quite rich and in many cases different from its classical counterpart. In particular, quantum-cryptographic primitives such as quantum pseudorandom states (PRS) \cite{PRandomQuantStates} could be used to achieve cryptographic functionality \cite{QuantCommitWithoutOWF, CryptoFromPRS} but are plausibly not implied by one-way functions, the most fundamental classical cryptographic primitive \cite{QuantPRandAndClassicComplexity, QuantCryptoInAlgorithmica}. The study of ``Microcrypt'', the cryptographic landscape that is not implied by one-way functions, has emerged as a central object of inquiry in the theory of quantum cryptography, see e.g.\ \cite{QuantCommitWithoutOWF,PRandomFunctionLikeStates, CompHardnessForQuantCrypto, PubKeyEncWithQKeys, 1WaynessQuantCrypto, CommitFromQuant1Way, Batra024,PrimForQCryptographyClassicaComm,MetaQuantCrypto}. In particular, it has become an important task to map out the different primitives and their interconnections.

This work focuses on the existential relation between various notions of PRS, which are an important part of Microcrypt. A PRS is a quantum state that can be efficiently generated (starting from a classical random seed), and is computationally indistinguishable from a truly random quantum state, sampled from the Haar-random distribution.\footnote{If we think of quantum states as unit vectors in a complex Hilbert space, then the Haar random distribution can be thought of as sampling a random vector over the unit sphere.} We note that we are only considering pure-state PRS in this work. The original \cite{PRandomQuantStates} definition required that the PRS is indistinguishable from random even when given an a-priori unbounded number of \emph{copies} of the state.\footnote{We recall that, in contrast to the classical setting, the more copies of the same quantum states that are given, the more information about that state can be extracted.} This in particular means that if the output state contains $\outqubits$ qubits, then any construction with $\poly(\outqubits)$-bit seed would require computational hardness (i.e.\ information theoretic security would be impossible). This would not be the case if we only required $t$-copy security for some parameter $t$. Namely, if we required that $t$ copies of the PRS are indistinguishable from $t$ copies of a Haar-random state. In particular, it is possible to achieve this notion information theoretically using so-called ``state-designs'' with a seed length of roughly $tm$ bits.

Despite the above, it turns out that studying these so-called $t$-PRS may prove quite instructive. In particular, there is much literature on the notion of $1$-PRS where only one copy of the state is given to the distinguisher. This object is convenient to work with, since when given just a single copy of the state, the Haar random distribution is identical to the classical uniform distribution. Naturally, it is cryptographically non-trivial to study this object in a setting where the seed length $\keylen$ is shorter than the output length $\outqubits$. Indeed, this notion of (cryptographic) $1$-PRS has been the focus of significant research \cite{QuantCommitWithoutOWF, CommitQuantStates, MetaQuantCrypto, SeperatQuantCommitAndQuant1Wayness}. In particular, it is known that $1$-PRS is at the ``very bottom'' of Microcrypt complexity, in the following sense. First, it is known to imply cryptographic functionality (such as commitment schemes). Second, it is not information theoretically possible, so an unbounded (quantum) attacker can distinguish any $1$-PRS from random. And third, it seems to reside outside the classical complexity landscape, in the sense that it is not known to be violated by a bounded quantum attacker with access to a completely computationally unbounded classical attacker \cite{1QueryUnitarySynth}. This is in contrast to \cite{PRandomQuantStates}-style PRS which are known to be violated given a $P^{\# P}$ oracle, and are known to imply more elaborate cryptographic tasks \cite{PRandomQuantStates, PRandomFunctionLikeStates}. Indeed, there are explicit separation results between $1$-PRS and PRS \cite{SeparatingQuantPRandomness}.

In light of the above, it seems instructive to study the middle ground. The setting of $t$-PRS where $t$ is some asymptotically increasing polynomial function, but is a-priori determined and is not up to the adversary. These objects received fairly little attention in the literature thus far and not much is known about their place within Microcrypt. It was shown in \cite[Theorem C.2]{CommitQuantStates} that a $t$-PRS implies a $1$-PRS with a longer output size, essentially corresponding to the entropy that can be extracted from $t$-copies of a Haar random state. Indeed, comparing to the entropy of a $t$-copy Haar-random state is the benchmark for comparing the key size.\footnote{One could also consider comparing against the randomness-complexity of a state $t$-design, but the numbers are very similar.} Since this entropy is roughly $t \outqubits$, it will be convenient to compare $t \outqubits$ to the key size $\keylen$, and we refer to the value $t \outqubits - \keylen$ as the ``stretch'' of the $t$-PRS.
One can also infer from \cite{CommitFromQuant1Way} that if $t \ge c \cdot \keylen$ for some global constant $c$, then a $t$-PRS implies a notion called one-way puzzles, which are not known to follow from a $1$-PRS. In light of this, one should not expect to amplify $1$-PRS into that domain. Note that in all of the above, it is not even clear if $1$-PRS implies a $2$-PRS and certainly not whether it is possible to support a number of copies that grows asymptotically with the security parameter.

Very recently, Ananth and Goldin \cite{LessIsMore} (henceforth, AG) addressed this question, and showed that if the $1$-PRS generator has a very special form, then it is possible to use it to achieve $t$-PRS for any $t$ that is polynomial in the security parameter. Their result produces a non-trivial (i.e.\ entropy-expanding) $t$-PRS only if the circuit $G$ that generates the $1$-PRS does not use too many ancilla qubits (as a function of the other parameters of the $1$-PRS). Therefore, this result still does not rule out even a separation between $1$-PRS and $2$-PRS in the general setting. In this work, we show that it is possible to amplify $1$-PRS to $t$-PRS generically, for any polynomial $t=t(\secp)$ in the security parameter. Thus we show for the first time that increasing the number of copies of a PRS does not yield a stronger object.

\subsection{Our Results}

We show that \emph{any} $1$-PRS can be amplified into a $t$-PRS for \emph{any} polynomial $t=t(\secp)$ (where $\secp$ is the security parameter). Slightly more formally, we show the following result.

\begin{theorem}[Informal]
Let $G$ be an efficient $1$-PRS generator with seed length $\keylen$ and output length $\outqubits$, and let $t=\poly(\secp)$. Then there exists an efficient $\tParam{G}$ which is a $t$-PRS with seed length $\tParam{\keylen} = t (\keylen+k'+O(\secp))$ and output length $\tParam{\outqubits}= \outqubits+k'$, where $k' \leq \poly(\secp,\abs{G})$ (where $\abs{G}$ is the size of the purified circuit that generates the $1$-PRS).
\end{theorem}

Let us try to explain the parameters of the theorem. Naively, we are ``cranking together'' $t$ instantiations of the $1$-PRS in order to obtain a $t$-PRS. So one would expect $\tParam{\keylen} \approx t \keylen$ and $\tParam{\outqubits} \approx \outqubits$. However, in the actual implementation, we need some additional randomness to make the process go through, but this randomness is not consumed and is retrieved in the output (this is similar to the use of a seed in a strong extractor, which is one of the tools that we use). So we get to $\tParam{\keylen} \approx t(\keylen+k')$ and $\tParam{\outqubits} \approx \outqubits+k'$. The dependence on $\abs{G}$ comes from depolarizing the ancilla qubits used in the execution of $G$. However, we have some ``parasitic costs'' that require an additional $O(\secp)$ per produced copy of the $t$-PRS, which leads to the expressions in the theorem.

Therefore, if we consider the ``stretch'' of our construction we get $t(\outqubits-\keylen-O(\secp))$. We recall that $\outqubits>\keylen$ since the original $1$-PRS is non-trivial. 
A naive application of this theorem in a setting where $\outqubits-\keylen$ is small (e.g.\ $\outqubits-\keylen=1$) might not be suitable for our amplification theorem. However, we notice that all we require is that the \emph{additive} stretch, i.e.\ $\outqubits-\keylen$, is greater than $O(\secp)$. This can be handled easily by first sequentially repeating the $1$-PRS. That is, it is straightforward to go from $1$-PRS with parameters $(\keylen,\outqubits)$ to $1$-PRS with parameters $(d \keylen, d \outqubits)$ for any $d=\poly(\secp)$, by concatenating $d$ instantiations side by side. Taking the appropriate value $d=O(\secp)$, we get a sufficient stretch to apply our amplification and obtain a valid construction.

We note that whereas we can support arbitrary polynomial dependence of $t$ in the security parameter, our construction increases the length of the seed as well. Therefore, it is not possible to amplify $t$ as a function of the (new) seed length $\tParam{\keylen}$. This is unavoidable if we accept the aforementioned separations. Once we have $t \ge c \cdot \keylen$, we get an inherently weaker primitive that we do not expect to construct from $1$-PRS.

\subsection{Technical Overview}

We start by recalling the basic idea of \cite{LessIsMore} (paraphrased for the purposes of this paper). Let $G$ be a purified circuit that generates the $1$-PRS with the following syntax: $G \ki\ket{0} = \kpi\kgi$, where $i \in \binset^\keylen$ is a seed value, $\kpi$ is the produced $\outqubits$-qubit state, and $\kgi$ is a garbage state. The generation circuit can always be presented in this way without loss of generality. Let us assume for a second that there is no garbage state, so it is possible to generate a superposition of the form
\[
\frac{1}{2^{\secp/2}}\sum_{i \in \binset^\secp} (-1)^{f_1(i)} \ki \kpx{f_2(i)}~,
\]
where $f_1$ is a random binary function and $f_2$ is a random function from $\binset^\secp$ to $\binset^\keylen$. Now if we take $t$ copies of this state, we get a state of the form:
\[
\frac{1}{2^{\secp t/2}}\sum_{\vi \in (\binset^\secp)^t} \underbrace{(-1)^{\sum_{j=1}^t f_1(i_j)}}_{\text{denote $(-1)^{f_1(\vi)}$}} \kvi \underbrace{\bigotimes_{j=1}^t \kpx{f_2(i_j)}}_{\text{denote $\kpx{f_2(\vi)}$}} = \frac{1}{2^{\secp t/2}} \sum_{\vi} (-1)^{f_1(\vi)} \kvi \kpx{f_2(\vi)} ~.
\]
An overwhelming fraction of the mass of this tensor product resides on vectors where $\vi$ contains $t$ distinct values. This is known as the ``distinct subspace'' and plays a very important role in many results having to do with quantum pseudorandomness. The random phase removes all correlations between $\vi, \vi'$ unless $\vi'$ is a permutation of $\vi$, we refer to this here as ``symmetric decoupling''. This is by now a standard technique in quantum pseudorandomness which is not new to our work or to AG, so we will not get into the details. An important note is that in AG, a larger-order root of unity was used instead of $(-1)$, which is wasteful and, as we show in our work, not required. We therefore just use $(-1)$ from the start here, to avoid clutter in the notation.

We may therefore assume from now on that $\vi$ only ranges over distinct values, and furthermore, each $\kvi$ only has correlations with its symmetric counterparts. Furthermore, we notice that when restricted to distinct $\vi$, the state $\kpx{f_2(\vi)}$ is just a $t$-tensor of $t$ independent instances of the $1$-PRS. Applying the $1$-PRS property, this is computationally indistinguishable from a $t$-tensor of independent Haar-random states. It holds that a symmetrically decoupled state containing a $t$-tensor of independent Haar-random states is close to a $t$-copy Haar random state.\footnote{This is not a contribution of our paper so we don't get into details, but at the level of ``footnote intuition'' we can explain that being symmetrically invariant means that the $t$-fold state behaves the same as $t$-copies of one state, and since each marginal is random, this is indeed similar to $t$-copies of a Haar random state.}

Therefore, the above simplified version of AG indeed produces a state that is $t$-copy indistinguishable from Haar, but it uses random functions $f_1, f_2$. This is where AG notice that it is possible to create the above $t$-tuple by making only $t$ oracle calls to $f_1, f_2$. They can therefore use the well known result by Zhandry \cite{Zah12} and replace $f_1, f_2$ by $2t$-wise independent functions. This means that the seed length of the construction is roughly $\tParam{\keylen} = 2t(\secp + \keylen)$ (assuming for simplicity $\secp \le \keylen$). The output length is $\tParam{\outqubits} = \secp + \outqubits$. Therefore, the new stretch that they get is $t \tParam{\outqubits} - \tParam{\keylen} = t(\outqubits - 2 \keylen - \secp)$. Therefore, even for this simple variant, one needs the initial $1$-PRS to be at least length-doubling in order to have a chance of getting non-trivial $t$-PRS. Note that the sequential composition technique discussed above will not help in this case.

In fact, the above is \emph{the most favorable setting} for the AG construction. Recall that we assumed that $\kgi$ is empty. A central technical challenge in AG is how to address the possibility of a non-empty $\kgi$. Their idea is to use quantum one-time pad (QOTP): to use randomness from the seed to completely depolarize $\kgi$, i.e.\ to ``encrypt'' it so that it is effectively removed from the state. QOTP asserts that applying $X^x Z^z$ for random $x,z\in\binset$ to any $1$-qubit state completely depolarizes the state of this qubit. Their final construction, therefore, is of the form
\[
\frac{1}{2^{\secp/2}}\sum_{i \in \binset^\secp} (-1)^{f_1(i)} \ki \kpx{f_4(i)} X^{f_2(i)} Z^{f_3(i)} \kgx{f_4(i)}~,
\]
where all $f_1, f_2, f_3, f_4$ are $2t$-wise independent as before (note that now $f_4$ plays the same role as $f_2$ in the simplified construction). The output length of $f_2, f_3$ is exactly the qubit-length of $\kgi$, which is the number of output ancilla qubits in the circuit $G$. We denote this number by $\glen$.

With this addition, the parameters they achieve are $\tParam{\keylen} = 2t(\secp + \keylen + 2\glen)$ and $\tParam{\outqubits} = \secp + \outqubits + \glen$. Now $t \tParam{\outqubits} - \tParam{\keylen} = t(\outqubits -2 \keylen - \glen - \secp)$, so it is not even enough that $\outqubits > 2 \keylen$, but it also needs to account for the $\glen$ qubits of the ancilla. Therefore their result is only applicable in a fairly narrow regime of parameters.

\paragraph{Our Improvements.} We would like to use the same components as AG, but ensure that we obtain a meaningful result in all parameter regimes. Conceptually, our techniques can be viewed as handling two artifacts separately. 

First, in order to handle the length-doubling constraint, we propose a tighter analysis of the AG approach, showing that the full strength of Zhandry's result is not required here. Indeed, whereas $f_1$ is required to be $2t$-wise independent, the state after symmetric decoupling is, well, symmetric (i.e., its density matrix contains only terms of the form $\ketbra{\vi}{\pi(\vi)}$, which involve at most $t$ distinct inputs), and therefore it suffices to take $f_2, f_3, f_4$ to only be $t$-wise independent. Therefore, for the simple variant without ancilla, we can obtain $t \tParam{\outqubits} - \tParam{\keylen} = t(\outqubits - \keylen - \secp)$, where we recall that $O(\secp)$ slackness can be handled by sequential composition.

Second, we need to handle the dependence on $\glen$. To this end, we use a similar technique to the one used by Cavalar et al.~\cite{MetaQuantCrypto} in the context of quantum meta-complexity. While their work is fairly technical, we believe that there is an important underlying intuitive insight. It is well established that QOTP requires a key that is twice the length of the state to be encrypted (the ``message state''). However, this is only really required if the message state is fully entangled with an adversarial environment (e.g.\ encrypting half of an EPR pair). Indeed, for a multi-qubit pure state, it suffices to use a secret key of (roughly) the message-length, assuming the existence of a common random string. To explain this, consider a maximally entangled state $(U_1\otimes U_2) \sum_x \ket{x} \ket{x}$, where $U_1, U_2$ are arbitrary. Then it suffices to QOTP encrypt only half of the qubits (which requires $2n$ bits, the same as the message length) in order to fully randomize the state. Viewed differently, it suffices to trace out half of the state in order for the remainder to become uniform, and the tracing out is implemented by a depolarizing QOTP. 

The idea in \cite{MetaQuantCrypto} is to use a \emph{quantum strong extractor}. Intuitively, applying an extractor scrambles a $2n$-qubit pure state so that the entanglement between the first and second half is roughly $n$, which allows to apply the above intuition. Concretely, if we apply a unitary $2$-design (e.g.\ a random Clifford) to a pure quantum state, the marginal distribution of the first $n-\log(1/\varepsilon)$ qubits becomes $\varepsilon$-close to maximally mixed. Therefore it suffices to depolarize the remaining $n+\log(1/\varepsilon)$ qubits, using a key of length $2(n+\log(1/\varepsilon))$ to achieve the required result. This insight can be generalized to non-pure states so long as the min-entropy of the state conditioned on the environment could be lower-bounded (recall that entanglement causes the conditional entropy to be negative). However, for our analysis the pure version suffices. Note that the strong extractor property means that the randomness used to generate the $2$-design can also be a part of the output.\footnote{We thank an anonymous TCC reviewer for pointing out that instead of using an extractor, we can potentially use the approximate quantum encryption scheme of Ambainis and Smith~\cite{AmbainisSmith04}. We believe that this can be made to work, but we have not verified all the details.}

Finally, we take $\varepsilon = 2^{-\secp}$ and obtain the following construction:
\[
\frac{1}{2^{\secp/2}}\sum_{i \in \binset^\secp} (-1)^{f_1(i)} \ki \ket{f_5(i)} \kpx{f_4(i)} X^{f_2(i)} Z^{f_3(i)} U_{f_5(i)} \kgx{f_4(i)}~,
\]
where $U$ is a family of unitary $2$-designs, $f_1$ is a $2t$-wise independent function with output length $1$, $f_2, f_3, f_4, f_5$ are $t$-wise independent with output lengths $\glen/2+\secp, \glen/2+\secp, \keylen, \designseedlen$, where $\designseedlen$ is the seed length of the $2$-design. Note that $X^{f_2(i)} Z^{f_3(i)}$ only act on the first $\glen/2+\secp$ qubits of the state $U_{f_5(i)} \kgx{f_4(i)}$. We therefore achieve $\tParam{\keylen} = 2t\secp + t(\keylen + \glen + 2 \secp + \designseedlen)$, $\tParam{\outqubits} = \secp + \outqubits + \glen + \designseedlen$. This finally implies that for our construction $t \tParam{\outqubits} - \tParam{\keylen} = t(\outqubits - \keylen -3\secp)$, and therefore, together with sequential repetition if needed, we obtain a $t$-PRS for any input $1$-PRS.

\section{Preliminaries}
\label{preliminaries}

\subsection{Quantum Information}
For a system of $m$ qubits residing in register $R$ we use $\abs{R}$ to denote the dimension of the system, that is, $\abs{R} = 2^m$. We now define properties of linear operators acting on quantum systems.

Let $\cH \cong \bbC^{2^m}$ be a Hilbert space over $m$ qubits, and $A$ a linear operator acting on $\cH$. We say that $A$ is a \emph{sub-normalized state} if it is PSD and $\tr[A] \leq 1$. The \emph{trace norm} of $A$ is defined by 
$$\trnorm{A} = \tr \left[ \sqrt{A^{\dagger} A} \right] = \text{sum of singular values of $A$} ~.$$
We say that $A$ is a \emph{quantum state} if it is sub-normalized and $\tr[A] = 1$. If $\rho, \sigma$ are sub-normalized then we write $\rho \preceq \sigma$ when $\sigma - \rho$ is PSD, i.e. $0 \preceq \sigma - \rho$. The \emph{trace distance} of quantum states $\rho, \sigma$ is defined by 
$$\td{\rho}{\sigma} = \frac{1}{2} \trnormdis{\rho}{\sigma} ~.$$

We move to definitions that consider a tensor product of $t$ identical Hilbert spaces.

\begin{definition}[The Distinct Set and The Distinct Subspace]
    Let $n, t \in \bbN$.
    \begin{itemize}
    \item The \emph{distinct set} is defined as
    $$\dist = \set{(i_1, \dots, i_t) \mid \forall j \neq k \quad i_j \neq i_k}
    \subseteq \left(\binset^n \right)^t ~.$$

    \item The \emph{distinct subspace} is the subspace spanned by the distinct set, i.e., $\operatorname{span} \set{\ket{\bv{i}} \mid \bv{i} \in \dist} \subseteq (\bbC^{2^n})^{\otimes t}$, where if $\bv{i} = (i_1, \dots, i_t)$ then $\ket{\bv{i}}$ is a shorthand for $\ket{i_1} \otimes \dots \otimes \ket{i_t}$.

    \item The \emph{projector onto the distinct subspace} is defined as 
    $$\distproj = \sum_{\bv{i} \in \dist} \ketbra{\bv{i}}{\bv{i}} ~.$$
    If a system consists of $t$ copies of registers $CE$ where $C$ is a register of $n$ qubits, then the projector onto the distinct subspace of $C$ is $\distprojreg{C} = I_E \otimes \sum_{\bv{i} \in \dist} \ketbra{\bv{i}}{\bv{i}}_C$.
    \end{itemize}
\end{definition}

\begin{definition}[Symmetric Subspace]
    Let $m,t \in \bbN$ and a Hilbert space over $m$ qubits, $\cH \cong \bbC^{2^m}$. Consider the Hilbert space $\cH' = \cH^{\otimes t}$.
    
    For any permutation $\pi \in S_t$ define the unitary permutation operator 
    \begin{flalign*}
        P_{\pi} 
        &= \sum_{x_1, ..., x_t \in \binset^m} \bigotimes_{j=1}^t \ketbra{x_j}{x_{\pi(j)}} \\
        &= \sum_{x_1, ..., x_t \in \binset^m} \ketbra{x_1, ..., x_t}{x_{\pi(1)}, ..., x_{\pi(t)}}
        ~.
    \end{flalign*}
    Define the \emph{symmetric subspace} to be
    $$\symparam{m}{t} = \set{ \ket{\psi} \in \cH' \mid P_{\pi} \ket{\psi} = \ket{\psi} \quad \forall \pi \in S_t} ~.$$
\end{definition}

We list some important properties of the symmetric subspace (see \cite{ChurchSymmetricSubspace} for a detailed discussion): 
\begin{itemize}
    \item The operator $\symprojparam{m}{t} = \frac{1}{t!} \sum_{\pi \in S_t} P_{\pi}$ is a projector onto the symmetric subspace. We denote the normalized projector by $\symrhoparam{m}{t} = \frac{\symprojparam{m}{t}}{\tr \left[ \symprojparam{m}{t} \right]}$.
    
    \item $\dim(\symparam{m}{t}) = \tr \left[ \symprojparam{m}{t} \right] = \binom{2^m + t - 1}{t}$. 

    \item $\symrhoparam{m}{t} = \E_{\ket{\psi} \gets \mu_m} \left[ \ketbra{\psi}{\psi}^{\otimes t} \right]$ where $\mu_m$ is the Haar random distribution of states over $m$ qubits.

    \item For any system of $m$ qubits with subsystem of $n$ qubits residing in register $C$, the projectors $\symprojparam{m}{t}$ and $\distprojreg{C}$ commute, i.e. $\symprojparam{m}{t} \distprojreg{C} = \distprojreg{C} \symprojparam{m}{t}$. 
\end{itemize}

We prove a simple claim that shows that projecting the normalized projector of the symmetric subspace of the system to the distinct subspace of a subsystem doesn't end up too far.
\begin{claim}
\label{proj_to_dis_subspace_close_to_sym}
    Let $R$ be a system of $r$ qubits and $C$ a system of $n$ qubits. Denote $p = n + r$, and define $\cH = (\cH_R \otimes \cH_C)^{\otimes t}$. Let $\symrhoparam{p}{t}$ be the normalized projector onto the symmetric subspace of $\cH$, and $\distprojreg{C}$ the projector onto the distinct subspace of $\cH_C^{\otimes t}$. If $\frac{2t}{2^n} \leq 1$ then $\trnormdis{\symrhoparam{p}{t}}{\symrhoparam{p}{t} \distprojreg{C}} \leq \frac{t^2}{2^n}$.
\end{claim}
\begin{proof}
    Since $\symrhoparam{p}{t}$ and $\distprojreg C$ commute, $\symrhoparam{p}{t}(I-\distprojreg C)$ is PSD, so
    \begin{flalign*}
        \trnormdis{\symrhoparam{p}{t}}{\symrhoparam{p}{t} \distprojreg{C}} 
        &= \trnorm{\symrhoparam{p}{t} (I - \distprojreg{C})} \\
        &= \tr \left[ \symrhoparam{p}{t} (I - \distprojreg{C}) \right] \\
        &= 1 - \tr \left[ \symrhoparam{p}{t} \distprojreg{C}\right]
    \end{flalign*}
    Observe that 
    $$\symrhoparam{p}{t} \distprojreg{C} = 
    \frac{1}{\binom{2^{n+r} + t - 1}{t}} \cdot \frac{1}{t!} \sum_{\pi \in S_t} \sum_{\bv{i} \in \dist, \bv{x} \in (\binset^r)^t} \bigotimes_{j=1}^t \ketbra{x_j}{x_{\pi(j)}}_R \otimes \ketbra{i_j}{i_{\pi(j)}}_C ~.$$
    For a fixed $t$-tuple $\bv{i} \in \dist$, index $j \in [t]$ and permutation $\pi \in S_t$, 
    $$\tr[\ketbra{i_j}{i_{\pi(j)}}_C] = \braket{i_j}{i_{\pi(j)}} = \begin{cases}
        1 & i_j = i_{\pi(j)} \\
        0 & i_j \neq i_{\pi(j)}
        ~.
    \end{cases}$$
    So the pair $\bv{i}, \pi$ contribute to the trace $\iff$ for all $j \in [t]$, $i_j = i_{\pi(j)}$. Since $\bv{i}$ is a tuple of distinct elements, this is equivalent to $\pi$ being the identity. Therefore,
    \begin{flalign*}
        \tr \left[ \symrhoparam{p}{t} \distprojreg{C} \right]
        = \frac{2^{rt} \abs{\dist}}{\binom{2^{n+r} + t - 1}{t} t!} 
        = \frac{2^{rt} \binom{2^n}{t}}{\binom{2^{n+r} + t - 1}{t}}
        ~.
    \end{flalign*}
    Therefore,
    \begin{flalign*}
        \trnormdis{\symrhoparam{p}{t}}{\symrhoparam{p}{t} \distprojreg{C}} 
        = 1 - \frac{\binom{2^n}{t} \cdot 2^{rt}}{\binom{2^{n+r} + t - 1}{t}}
        = 1 - \prod_{j=0}^{t-1} \left( 1 - \underbrace{\frac{j (2^r + 1)}{2^{n+r} + j}}_{\epsilon_j} \right)
        ~.
    \end{flalign*}
    For each $j$, 
    $$\epsilon_j = \frac{j (2^r + 1)}{2^{n+r} + j} \leq \frac{2 j \cdot 2^r}{2^{n+r}} = \frac{2j}{2^n} \coloneqq \epsilon'_j ~.$$
    By assumption, $\frac{2t}{2^n} \leq 1$, so also $\epsilon_j \leq \epsilon'_j \leq 1$ for all $j$. Hence $\trnormdis{\symrhoparam{p}{t}}{\symrhoparam{p}{t} \distprojreg{C}} \leq 1 - \prod_{j=0}^{t-1} (1 - \epsilon'_j)$ and we can use the union bound $1 - \prod_j (1 - \epsilon'_j) \leq \sum_j \epsilon'_j$ and get
    \begin{flalign*}
        \trnormdis{\symrhoparam{p}{t}}{\symrhoparam{p}{t} \distprojreg{C}}
        \leq \sum_{j=0}^{t-1} \frac{2j}{2^n} 
        = \frac{2 t (t - 1)}{2 \cdot 2^n} 
        \leq \frac{t^2}{2^n}
        ~.
    \end{flalign*}
\end{proof}

\subsection{Pseudorandom States}

A quantum polynomial time (QPT) algorithm is a sequence of quantum circuits $C = \set{C_\secp}_\secp$ which are polynomially bounded, that is, there exists a polynomial $p(\cdot)$ such that the circuits $C_\secp$ are of size at most $p(\secp)$.\footnote{Note that this definition is non-uniform but a uniform version can be defined analogously in a straightforward manner.} A \emph{unitary} quantum algorithm is a sequence of quantum circuits such that each $C_\secp$ is a unitary mapping. If the output of a quantum algorithm is a classical bit, we call it a distinguisher (or distinguishing adversary). 

For such a distinguisher we define acceptance probability on a sequence of sub-normalized states $\rho = \set{\rho_{\secp}}_{\secp}$, denoted $\Pr[C_\secp(\rho_\secp)=1]$, to be $0$ if $\tr [\rho_\secp] = 0$ and otherwise $\Pr \left[ C_\secp \left(\frac{\rho_\secp}{\trnorm{\rho_\secp}} \right) = 1 \right] \cdot \tr[\rho_\secp]$. With this convention, the usual definitions of statistical and computational indistinguishability of quantum state ensembles extend directly to sub-normalized states. We use the notation $\cind$ to denote that states are computationally indistinguishable, and $\sind$ to denote that states are statistically indistinguishable.

\begin{proposition}
Let $\rho = \set{\rho_{\secp}}_{\secp}$  and $\sigma = \set{\sigma_{\secp}}_{\secp}$ be ensembles of sub-normalized states of the same dimensions, and $\varepsilon$ a negligible function. If for all $\secp$, $\trnormdis{\rho_\secp}{\sigma_\secp} \leq \varepsilon(\secp)$ then any quantum distinguisher $\cA = \set{\cA_\secp}_\secp$ can distinguish $\rho$ from $\sigma$ with advantage at most $\varepsilon$, i.e. 
$$\abs{ \Pr[\cA_\secp(\rho_\secp)=1]
- \Pr[\cA_\secp(\sigma_\secp)=1]} \leq \varepsilon(\secp) ~.$$
Moreover, if for all $\secp$, $\rho_\secp$ and $\sigma_\secp$ are normalized and $\td{\rho_\secp}{\sigma_\secp} \leq \varepsilon(\secp)$ then also in this case $\cA$ can distinguish $\rho$ from $\sigma$ with advantage at most $\varepsilon$.
\end{proposition}
Note that the above proposition is true also for distinguishers that are not polynomially bounded, and in fact we get statistical indistinguishability which implies computational indistinguishability.

We now give a definition of a keyed procedure that generates a pure quantum state. This will allow us to argue about the cryptographic properties of the output.
\begin{definition}[Pure State Generator]
	A $(\keylen, \outqubits)$-pure-state ensemble is a set of $\outqubits$-qubit pure quantum states, indexed by a $\keylen$-bit classical key: $S = \set{\ket{\varphi_i}}_i$. A unitary quantum algorithm $G$ is a (pure-state) generator for the ensemble $S$ if on input $i$ it outputs $\ket{\varphi_i}$, more explicitly, if $G \ket{i} \ket{0^l} = \ket{\varphi_i} \ket{g_i}$, where $\ket{g_i}$ is an $\garbqubits$-qubit ``garbage output'' to be discarded. 
	
	Asymptotically, given a sequence of ensembles $\set{S_\secp}_{\secp}$, where $\secp \in \bbN$ is the security parameter, a unitary QPT algorithm $G=\set{G_\secp}_{\secp}$ is a \emph{pure state generator} for this sequence if $G_\secp$ generates $S_\secp$ for all $\secp$.
\end{definition}

We can now define pseudorandom states which are generated by a pure state generator and have randomness properties.
 
\begin{definition}[Pseudorandom State]
\label{p_def_t_prsg}
    A sequence of ensembles $S = \set{S_\secp}_\secp$ generated by a pure state generator, where $\secp$ is the security parameter, is a \emph{$(\keylen, \outqubits, t)$-pseudorandom state} if $\keylen, \outqubits, t$ are polynomials in $\secp$ such that $S_\secp = \set{\ket{\varphi_{\secp, i}}}_i$ is a $(\keylen(\secp), \outqubits(\secp))$-pure-state ensemble for all $\secp$, and for any QPT distinguishing adversary $\cA = \set{\cA_\secp}_\secp$,
    \begin{flalign*}
        \abs{
        \Pr_{i \gets \binset^{\keylen(\secp)}} \left[ \cA_\secp \left( 
        \ket{\varphi_{\secp, i}}^{\otimes t(\secp)} \right) = 1 \right]
        -
        \Pr_{\ket{\psi} \gets \mu_{\outqubits(\secp)}} \left[ \cA_\secp \left( \ket{\psi}^{\otimes t(\secp)} \right) = 1 \right]
        } \leq \negl(\secp)
        ~,
    \end{flalign*}
    where $\mu_{\outqubits(\secp)}$ is the Haar distribution over $\outqubits(\secp)$-qubit states.
\end{definition}

We frequently write $t$-PRS instead of $(\keylen, \outqubits, t)$-PRS when the parameters $\keylen, \outqubits$ are either arbitrary or clear from the context.
As explained in the introduction, we focus on the non-trivial regime where $\keylen < t \outqubits$, and if a $t$-PRS satisfies this condition we say it has a \emph{stretch} of $t \outqubits - \keylen$.

\begin{claim}
\label{prop_1prs_amplification}
    If there exists an $(\keylen, \outqubits, 1)$-PRS with stretch $\s = \outqubits - \keylen > 0$,  where $\secp$ is the security parameter, then for any polynomial $d \coloneqq d(\secp) > 0$ there exists a $(\ampParam{\keylen}, \ampParam{\outqubits}, 1)$-PRS for $\ampParam{\keylen} = d \keylen$, $\ampParam{\outqubits} = d \outqubits$, with stretch $\ampParam{\s} = d \s$.
\end{claim}
\begin{proof}
    The construction is by concatenation. Take $d$ independent copies of the original $1$-PRS $\set{\ket{\varphi_i}}_i$, i.e.\ parse the key $\ampParam{\keylen}$ as $d$ keys $i_j$ of length $\keylen$ and output $\ket{\varphi_{i_1}}\otimes \dots \otimes \ket{\varphi_{i_d}}$. The ``stretch'' calculation is straightforward. Security follows by a hybrid argument from the security of the original PRS, by replacing the states $\ket{\varphi_{i_j}}$ with $\ket{y_j}$ for random $y_j$, one by one. This is possible since for a random $\ampParam{\keylen} = (i_1, \dots, i_d)$ the $i_j$s are random and independent from one another. Furthermore, the density matrix of a single-copy Haar-random state over $\ampParam{\outqubits}$ qubits is maximally mixed and can be written as a product of $d$ maximally mixed states over $\outqubits$ qubits.

\end{proof}

\subsection{Quantum Extractor}
We first present definitions of min-entropy that will be useful for applying quantum extractors.
\begin{definition}[Conditional min-entropy]
\label{p_cond_mentropy}
    Let $\rho_{BE} \in \cH_B \otimes \cH_E$ be a sub-normalized state on subsystems $B$ and $E$. The \emph{conditional min-entropy} of $\rho_{BE}$ is defined as
    $$H_{\infty}(B|E) = H_{\infty}(B|E)_{\rho}  = \sup_{\sigma_E} \set{\lambda \mid \rho_{BE} \preceq 2^{-\lambda} I_B \otimes \sigma_E} ~,$$
    where $\sigma_E$ is any sub-normalized state over system $E$. 
\end{definition}

\begin{definition}[Conditional Smoothed min-entropy]
\label{p_cond_smooth_mentropy}
    Let $\rho_{BE} \in \cH_B \otimes \cH_E$ be a sub-normalized state on subsystems $B$ and $E$, and $\delta > 0$. The \emph{conditional smoothed min-entropy} of $\rho_{BE}$ is defined as
    $$H^{\delta}_{\infty}(B|E) = H^{\delta}_{\infty}(B|E)_{\rho}  = \sup_{\sigma_{BE} \in \cB^{\delta}(\rho)} H_{\infty}(B|E)_{\sigma} ~,$$
    where $\cB^{\delta}(\rho)$ is the ball of radius $\delta$ centered at $\rho_{BE}$, containing sub-normalized states $\sigma_{BE}$ over the system $BE$. This ball is defined under some metric, typically the Purified Distance $P$.\footnote{$P(\rho, \sigma) = \sqrt{1 - F(\rho, \sigma)^2}$ where $F(\rho, \sigma) = \trnorm{\sqrt{\rho} \sqrt{\sigma}}$ is the fidelity function. Also note that $\td{\rho}{\sigma} \leq P(\rho, \sigma)$.}
\end{definition}

We define quantum strong extractors and state a theorem that enables the use of $2$-designs as quantum strong extractors under certain conditions.
\begin{definition}[Quantum Strong Extractor, \cite{ClassicRandExtractors}]
    Let $\ell \in \bbN$ and $B = B_1 B_2$ be a quantum system with $B_1$ and $B_2$ as subsystems, where the subsystem $B_1$ consists of $\ell$ qubits. A collection of quantum unitaries $\set{U^r}_{r \in R}$ acting on system $B$ is called a \emph{$(\kExt, \epsExt, \delta)$-quantum strong extractor} that extracts $\ell$ qubits if for any quantum state $\rho_{BE} \in \cH_B \otimes \cH_E$ with $H^{\delta}_{\infty}(B|E) \geq \kExt$,
    $$\td{
    \frac{1}{|R|} \sum_{r \in R} \ketbra{r}{r} \otimes \tr_{B_2} \left( U^r \rho_{BE} \ {U^r}^{\dagger} \right)
    }{
    \frac{I_R}{|R|} \otimes \frac{I_{B_1}}{|B_1|} \otimes \rho_E
    }
    \leq \epsExt
    ~.$$
\end{definition}

\begin{theorem}[Quantum Strong Extractor from Unitary $2$-Designs \cite{DecouplingWithUnitary2Designs, ClassicRandExtractors, OneShotDecoupling, MetaQuantCrypto}]
\label{p_extractor_thm}
    Let $\nExt, \ell \in \bbN$, $\kExt \in [-\nExt, \nExt]$, and $\epsExt \in (0, 1)$ such that $\ell \leq \frac{\nExt + \kExt}{2} - \log \left( \frac{1}{\epsExt} \right)$. Then any unitary $2$-design on an $\nExt$-qubit system is a $(\kExt, \epsExt, \epsExt/12)$-quantum strong extractor that extracts $\ell$ qubits.
\end{theorem}

\subsection{Simulating symmetrization of states}
\label{symmetrization_simulator}
We introduce a symmetrization simulator that takes $t$ pure states as input and outputs a pure state representing a symmetrization of the input, entangled with a random $t$-tuple of distinct elements. While this simulator was originally introduced by \cite{LessIsMore}, we use a slightly modified variant that restricts the range of the $t$-tuple. Furthermore, we analyze the output density matrix of the simulator, which is later used to reduce $t$-PRS security to $1$-PRS security. 

\paragraph{Symmetrization Simulator.} 
Given $\ket{\phi_1}, \ldots, \ket{\phi_t}$, consider the following algorithm: Sample a random distinct $t$-tuple $\bv{i} \gets \dist$ and output the quantum state
$$
\tsimi{\bv{i}}{\ket{\phi_1}, \ldots, \ket{\phi_t}} \coloneqq 
\frac{1}{\sqrt{t!}}\sum_{\pi \in S_t} \bigotimes_{j=1}^t \ket{i_{\pi(j)}}\ket{\phi_{\pi(j)}}
~.$$

\begin{claim}
    The simulator algorithm described is QPT. 
\end{claim}
\begin{proof}
	Sample a random distinct $\bv{i}$. Then, in an auxiliary register, take a uniform superposition over all permutations $\frac{1}{\sqrt{t!}}\sum_{\pi \in S_t} \ket{\pi}$, and use it to compute the state 
    $$
    \frac{1}{\sqrt{t!}}\sum_{\pi \in S_t} \ket{\pi} \bigotimes_{j=1}^t \ket{i_{\pi(j)}} \ket{\phi_{\pi(j)}}~.
    $$
    Then uncompute $\pi$ using $\bv{i}, \bv{{\pi(i)}}$ (since $\bv{i}$ is distinct, given $\bv{i}$ and its permuted version allows to compute $\pi$ and therefore to uncompute the auxiliary register).
\end{proof}

\begin{claim}
\label{simulator_output}
    Given input $\ket{\phi_1}, \ldots, \ket{\phi_t}$, the output density matrix of the simulator is
    $$\frac{1}{\abs{\dist}} \sum_{\bv{i} \in \dist, \pi \in S_t} \bigotimes_{j=1}^t 
    \left( \ketbra{i_{j}}{i_{j}} \otimes \ketbra{\phi_{j}}{\phi_{j}} \right) P_{\pi} ~.
    $$
\end{claim}
\begin{proof}
    Denote the output density matrix by $\tsim{\ket{\phi_1}, \ldots, \ket{\phi_t}}$. By definition,
    \begin{flalign*}
        \tsim{\ket{\phi_1}, \ldots, \ket{\phi_t}} &=
        \E_{\bv{i} \gets \dist} \left[ \tsimi{\bv{i}}{\ket{\phi_1}, \ldots, \ket{\phi_t}} \ \tsimi{\bv{i}}{\ket{\phi_1}, \ldots, \ket{\phi_t}}^{\dagger} \right] \\
        &= \frac{1}{\abs{\dist}} \sum_{\bv{i} \in \dist}
        \frac{1}{t!} \sum_{\pi, \pi' \in S_t}
        \bigotimes_{j=1}^t \ketbra{i_{\pi'(j)}}{i_{\pi(j)}} \otimes \ketbra{\phi_{\pi'(j)}}{\phi_{\pi(j)}}
        ~.
    \end{flalign*}
Fixing a $t$-tuple $\bv{i}$ and permutation $\pi$, notice that the term $\bigotimes_{j=1}^t \ketbra{i_j}{i_{\pi(j)}} \otimes \ketbra{\phi_j}{\phi_{\pi(j)}}$ appears exactly $t!$ times in the summation, exactly once for each $\bv{i'}$ that is a permutation of $\bv{i}$, that is, $\exists \pi'$ such that $\pi'(\bv{i'}) = \bv{i}$. This determines the second permutation $\pi''$ such that $\pi''(\bv{i'}) = \pi(\bv{i}) = \pi(\pi'(\bv{i}))$. Therefore, summation over $\bv{i}, \pi, \pi'$ collapses to summation over $\bv{i}, \pi$ and the fraction $\frac{1}{t!}$ cancels out.
\begin{flalign*}
    \tsim{\ket{\phi_1}, \ldots, \ket{\phi_t}} &=
    \frac{1}{\abs{\dist}} \sum_{\bv{i} \in \dist}
    \frac{1}{t!} \sum_{\pi, \pi' \in S_t}
    \bigotimes_{j=1}^t \ketbra{i_{\pi'(j)}}{i_{\pi(j)}} \otimes \ketbra{\phi_{\pi'(j)}}{\phi_{\pi(j)}} \\
    &= \frac{1}{\abs{\dist}} \sum_{\bv{i} \in \dist}
    \sum_{\pi \in S_t}
    \bigotimes_{j=1}^t \ketbra{i_j}{i_{\pi(j)}} \otimes \ketbra{\phi_j}{\phi_{\pi(j)}} \\
    &= \frac{1}{\abs{\dist}} \sum_{\bv{i} \in \dist, \pi \in S_t} \bigotimes_{j=1}^t 
    \left( \ketbra{i_{j}}{i_{j}} \otimes \ketbra{\phi_{j}}{\phi_{j}} \right) P_{\pi}
\end{flalign*}
\end{proof}

\subsection{Auxiliary Claims}
For a seeded function family $\cF$, we denote $f \gets \cF$ for sampling a random seed for a function from the family, and associate $f$ both with the seed itself and with the function it defines. We denote by $\ell_{\cF}$ the seed length for sampling  a function from $\cF$. 
We state some useful claims that will help us prove that our construction is a $t$-PRS.

\begin{proposition}
\label{prop_indp_expect}
    Let $\cF \subseteq \cX \to \cY$ be a $t$-wise independent function family, and let $\set{\rho_y}_{y \in \cY}$ be a family of sub-normalized states. Then for every $t$ \emph{distinct} inputs $i_1, \dots, i_t \in \cX$,
    $$
    \E_{f \gets \cF} \left[\bigotimes_{j=1}^t \rho_{f(i_j)}\right]
    =
    \bigotimes_{j=1}^t \E_{f \gets \cF} \left[\rho_{f(i_j)}\right]
    ~.$$
\end{proposition}

\begin{proposition}
\label{prop_distance_pairs_tensor}
    For any two families of quantum states $\set{\rho_i}_i$, $\set{\sigma_i}_i$, if $\td{\rho_i}{\sigma_i} \leq \varepsilon$ for each $i$ then for any $t$, 
    $$\td{\bigotimes_{i=1}^t\rho_i}{\bigotimes_{i=1}^t\sigma_i}
    \leq t \varepsilon ~.$$
\end{proposition}

\begin{proposition}
\label{prop_norm_right_mult}
    Let $\rho,\sigma$ be sub-normalized states, and let $U$ be a unitary operator. If 
    $\trnormdis{\rho}{\sigma} \leq \varepsilon$, then $\trnormdis{\rho U}{\sigma U} \leq \varepsilon$.
\end{proposition}

\begin{proposition}
\label{prop_distance_expect}
    Let $\cD$ be some distribution, and $\set{\rho_d}_{d \in \cD}$, $\set{\sigma_d}_{d \in \cD}$ two families of states over this distribution. Then 
    $$\trnormdis{\E_{d \gets \cD}[\rho_d]}{\E_{d \gets \cD}[\sigma_d]} \leq \E_{d \gets \cD}[\trnormdis{\rho_d}{\sigma_d}] ~.$$
    In particular, if for all $d$ it holds that $\trnormdis{\rho_d}{\sigma_d} \leq \varepsilon$, then $\trnormdis{\E_{d \gets \cD}[\rho_d]}{\E_{d \gets \cD}[\sigma_d]} \leq \varepsilon$.
\end{proposition}

\begin{proposition}
\label{prop_otp}
    Let $\rho$ be a quantum state on register $A$, and assume $A$ is split into two registers $A = A_1 A_2$ such that $A_2$ holds $q$ qubits. Then $$\E_{x \gets \binset^q, z \gets \binset^q} \left[ 
    (I_{A_1} \otimes X^x_{A_2} Z^z_{A_2}) 
    \ \rho \
    (I_{A_1} \otimes Z^z_{A_2} X^x_{A_2})
    \right]
    = \tr_{A_2}(\rho) \otimes \frac{I_{A_2}}{\abs{I_{A_2}}}
    ~.$$
\end{proposition}

\begin{proposition}
\label{mult_1_plus_eps_i}
Let $\epsilon_1, \dots, \epsilon_t > 0$ and $\epsilon = \sum_{j=1}^t \epsilon_j$ such that $\epsilon <1$. Then $\prod_{j=1}^t (1+\epsilon_j) \le e^{\epsilon} \le 1+\epsilon+\epsilon^2$.
\end{proposition}

\begin{theorem}[{\cite[Corollary 3.34]{Pseudorandomness}}]
\label{t_family_existence}
    For any $t, \lambda, n \in \bbN$ there is a family of $t$-wise independent functions $\cF \subseteq \binset^{\lambda} \to \binset^n$ such that choosing a random function from $\cF$ takes $t \cdot \max \set{\lambda, n}$ random bits, that is, $\ell_{\cF} = t \cdot \max \set{\lambda, n}$. Moreover, evaluating any function from $\cF$ takes time $poly(n,\lambda, t)$.
\end{theorem}

\section{Constructing \texorpdfstring{$t$}{t}-PRS from \texorpdfstring{$1$}{1}-PRS}
\label{t_prs_construction}

\subsection{Parameter Setup}
Let $\designseedlen, \keylen, \outqubits, \garbqubits, n, t \in \bbN$.
Let $G$ be a pure state generator for an ensemble $\set{\ket{\varphi_i}}_i$ with key length $\keylen$, output register $A$ of $\outqubits$ qubits, and a garbage register $B$ of $\garbqubits$ qubits, such that for each key $i$, $G(i) = \ket{\varphi_i}_A \ket{g_i}_B$. Let $\epsExt \in (0, 1)$ be such that $\ell = \floor{\frac{\garbqubits}{2} - \log \left( \frac{1}{\epsExt} \right)} \geq 0$. Let $\cU = \set{U^r}_r$ be an efficiently implementable unitary $2$-design on the $\garbqubits$-qubit register $B$ with seed length $\designseedlen$. Denote by $B_2$ the register containing the last $q \coloneqq \garbqubits - \ell$ qubits of system $B$. 

Let $\cF_1 \subseteq \binset^n \to \binset$ be a $2t$-wise independent function family and 
$\cF_2 \subseteq \binset^n \to \binset^q$,
$\cF_3 \subseteq \binset^n \to \binset^q$,
$\cF_4 \subseteq \binset^n \to \binset^\keylen$,
$\cF_5 \subseteq\binset^n \to \binset^\designseedlen$ 
all $t$-wise independent function families.

\subsection{Construction}
Let $f_i \in \cF_i$, and define the following state.
\begin{flalign*}
    \ket{\psi_{f_1, f_2, f_3, f_4, f_5}}
    &= \frac{1}{\sqrt{2^n}}
    \sum_{i \in \binset^n}
    (-1)^{f_1(i)}
    \ket{f_5(i)}_R
    \left( X_{B_2}^{f_2(i)} Z_{B_2}^{f_3(i)} \right)
    U^{f_5(i)}_B
    \, G(f_4(i))_{AB} \,
    \ket{i}_C \\
    &= \frac{1}{\sqrt{2^n}}
    \sum_{i \in \binset^n}
    (-1)^{f_1(i)}
    \ket{f_5(i)}_R
    \left( X_{B_2}^{f_2(i)} Z_{B_2}^{f_3(i)} \right)
    U^{f_5(i)}_B
    \ket{\varphi_{f_4(i)}}_A \ket{g_{f_4(i)}}_B \ket{i}_C
\end{flalign*}

\subsection{Our Amplification Theorem}
\begin{theorem}[Main Theorem]
\label{main_thm}
    Let $\secp$ be a security parameter, and for each $\secp$ define the setup parameters w.r.t $\secp$ such that $G$ is a $1$-PRS generator, $\epsExt$ is negligible, $\ell \geq 0$, and $\designseedlen, \keylen, \outqubits, \garbqubits, n, t$ are polynomially bounded such that $n = \omega(\log(\secp))$. Define the pure state generator $\tParam{G}$ that outputs the construction above, namely for key $z = (f_1, f_2, f_3, f_4, f_5)$ outputs the pure state 
    $$\tParam{G}(z) = \ket{\psi_{f_1, f_2, f_3, f_4, f_5}} ~.$$
    
    Then $\tParam{G}$ is a $t$-PRS with key length 
    $$\tParam{\keylen} = t \left(
    2n + 2 \max \set{n, \ceil{\frac{\garbqubits}{2} + \log \left(\frac{1}{\epsExt}\right)}} + \max \set{n, \keylen} + \max \set{n, \designseedlen}
    \right) ~,$$
    and output length 
    $\tParam{\outqubits} = \designseedlen + \outqubits + \garbqubits + n$.
\end{theorem}

Beyond security, our construction preserves the ``stretch'' property of the $1$-PRS relative to $t$. That is, a $1$-PRS with stretch $s$ yields a $t$-PRS with stretch at least $ts$. In fact, by polynomially amplifying the stretch of the $1$-PRS, it is possible to construct a $t$-PRS with any stretch $ts$ for polynomially bounded $s$. 

\begin{corollary}
\label{cor_stretch_of_tprs}
    Let $\secp$ be a security parameter and $t$ a polynomial in $\secp$. If there exists a non-trivial $1$-PRS then there exists a non-trivial $t$-PRS with stretch $\geq t \s$ for any polynomial $\s \coloneqq \s(\secp)$.
\end{corollary}

\begin{proof} 
    Let $G$ be the state generator of a non-trivial $1$-PRS with stretch $s_G$. Using Claim~\ref{prop_1prs_amplification} we can instantiate another non-trivial $1$-PRS with generator $\ampParam{G}$ and stretch $d \cdot s_G$ for some $d$ which we will determine later. We then use $\ampParam{G}$ as the base $1$-PRS generator for the construction described in Theorem~\ref{main_thm} together with function families $\cF_1, \dots, \cF_5$ as guaranteed by Theorem~\ref{t_family_existence}, resulting in a new $t$-PRS generator $\tParam{G}$.
    Set $n = \secp$, $\epsExt = 2^{-\secp}$. We must make sure that $\ell = \floor{ \frac{\ampParam{\garbqubits}}{2} - \log \left( \frac{1}{\epsExt} \right) } > 0$, which can be enforced by adding extra ancilla qubits set to $\ket{0}$ if necessary.
    
    Using these parameters in Equation~\ref{size_of_seed_spaces} gives\footnote{W.l.o.g. we assume that $\designseedlen \geq \secp$, as the seed length of the $2$-design can be trivially extended. Similarly, assume that $\ampParam{\garbqubits}$ is even.}
    $$\tParam{\keylen} = t \left(
    2 \secp + \ampParam{\garbqubits} + 2 \log \left(\frac{1}{\epsExt}\right) + \max \set{\ampParam{\keylen}, \secp} + \designseedlen
    \right) ~.$$
    The output of $t$ copies is $t \tParam{\outqubits} = t (\designseedlen + \ampParam{\outqubits} + \ampParam{\garbqubits} + \secp)$ qubits. 
    The stretch is
    \begin{flalign*}
        t \tParam{\outqubits} - \tParam{\keylen} 
        &= 
        t (\designseedlen + \ampParam{\outqubits} + \ampParam{\garbqubits} + \secp)
        - 
        \left[ 
        t \left(
        2 \secp + \ampParam{\garbqubits} + 2 \log \left(\frac{1}{\epsExt}\right) + \max \set{\ampParam{\keylen}, \secp} + \designseedlen
        \right)
        \right] \\
        &= t \left(
        \ampParam{\outqubits}  - \max \set{\ampParam{\keylen}, \secp} - \secp - 2 \log \left(\frac{1}{\epsExt}\right) 
        \right) \\
        &= t \left(
        \ampParam{\outqubits}  - \max \set{\ampParam{\keylen}, \secp} - 3 \secp
        \right) 
    \end{flalign*}
    To ensure this value is at least $t \s$ we want $\ampParam{\outqubits} - \ampParam{\keylen} - 3 \secp \geq \s$ $\implies d \cdot s_G - 3 \secp \geq \s$. So it suffices to take $d = \ceil{\frac{\s + 3 \secp}{s_G}} = \poly(\secp)$, resulting in\footnote{W.l.o.g. we can assume $\ampParam{\keylen} \geq \secp$, since we can always take $d \geq \secp$.}
    \begin{flalign*}
        t \tParam{\outqubits} - \tParam{\keylen} &= t \left(
        \ampParam{\outqubits}  - \max \set{\ampParam{\keylen}, \secp} - 3 \secp
        \right) \\
        &= t \left(
        \ampParam{\outqubits}  - \ampParam{\keylen} - 3 \secp
        \right) \\
        &\geq t \s
    \end{flalign*}
\end{proof}

\subsection{Proof of Theorem~\ref{main_thm}}
\paragraph{Key length.} The key of $\tParam{G}$ consists of all the seeds to the function families $\cF_1, \dots, \cF_5$. By Theorem~\ref{t_family_existence}, 
\begin{equation}
\label{size_of_seed_spaces}
\begin{gathered}
    \ell_{\cF_1} = 2t \cdot n, \\
    \ell_{\cF_2} = \ell_{\cF_3} = t \cdot \max \set{n, q} = t \max \set{n, \garbqubits - \floor{ \frac{\garbqubits}{2} - \log \left(\frac{1}{\epsExt} \right) }} = t \max \set{n, \ceil{ \frac{\garbqubits}{2} + \log \left(\frac{1}{\epsExt} \right) }},
    \\ 
    \ell_{\cF_4} = t \cdot \max \set{n, \keylen}, \quad
    \ell_{\cF_5} = t \cdot \max \set{n, \designseedlen}
\end{gathered}
\end{equation}
So the total length of the key is 
$$\tParam{\keylen} = t \left(
2n + 2 \max \set{n, \ceil{ \frac{\garbqubits}{2} + \log \left(\frac{1}{\epsExt} \right) }} + \max \set{n, \keylen} + \max \set{n, \designseedlen}
\right) ~.$$

\paragraph{Output length.} The output length is simply the total number of qubits in registers $R, A, B, C$, which is 
$$\tParam{\outqubits} = \kappa + \outqubits + \garbqubits + n ~.$$

\paragraph{The Generator is QPT.}
Theorem~\ref{t_family_existence} ensures that all the functions $f_1, \dots, f_5$ can be evaluated efficiently, and the $2$-design as well as $G$ are efficiently computable. Therefore $\tParam{G}$ also has an efficient implementation which includes preparing a uniform superposition over $i$, applying the gates of $G$, and using controlled-$X$, controlled-$Z$, controlled-$U$ (for the $2$-design) and controlled-phase gates.

\paragraph{Our Hybrids.}
We turn to proving the security of our construction using a sequence of hybrids. Our goal is to show that $t$ copies of a state generated by $\tParam{G}$ is computationally indistinguishable from $t$ copies of a Haar-random state over the same dimensions.

We introduce a sequence of hybrids $\hbd_i$, where each hybrid is simply a density matrix. The hybrid $\hbd_0$ will denote a $t$-wise application of the generator $\tParam{G}$ on a random seed. The final hybrid $\hbd_4$ will be density matrix of a $t$-copy Haar-random state. We will show that $\hbd_0 \cind \hbd_4$ by the following outline:

\begin{itemize}
    \item The hybrid $\hbd_0$ is just the $t$-wise density matrix of our construction. 

    \item The hybrid $\hbd_1$ is a restriction to the so-called distinct subspace.

    We show that $\hbd_0 \sind \hbd_1$ by showing that most of the mass of our density matrix is concentrated in the distinct subspace.

    \item The hybrid $\hbd_2$ is is obtained from the hybrid $\hbd_1$ by ``depolarizing'' the ancilla register ($B$) as well as the extractor-key register ($R$), that is, replacing the contents of registers $R, B$ with maximally mixed states.

    We show that $\hbd_1 \sind \hbd_2$ by showing that the application of the extractor followed by the quantum one-time pad in register $B$ of the hybrid $\hbd_1$ (with the key of the extractor stored in register $R$), mixes the states in these registers such that they look almost uniformly random.

    \item The hybrid $\hbd_3$ is obtained from the hybrid $\hbd_2$ by replacing each of the $1$-PRS states (in register $A$) with a maximally mixed state.

    Using the security of the $1$-PRS, we show that $\hbd_2 \cind \hbd_3$.

    \item The hybrid $\hbd_4$ is the density matrix of a $t$-copy Haar-random state.

    We show that $\hbd_3 \sind \hbd_4$ by showing that the
    hybrid $\hbd_3$ is a sub-normalized state proportional to the projection onto the intersection of the symmetric and distinct subspaces, and captures most of the mass of the  normalized projector onto the symmetric subspace, which coincides with hybrid $\hbd_4$.
\end{itemize}

Combining all these Hybrids gives
$$
\hbd_0
\sind
\hbd_1
\sind
\hbd_2
\cind
\hbd_3 
\sind 
\hbd_4
~,
$$
which completes the security proof of the $t$-PRS.

\paragraph{Formal Hybrid Definitions and Claims.}
We define $\hbd_0$ as 
\begin{equation}
\begin{aligned}[t]
    \hbd_0 &\coloneqq 
    \E_{f_1, \dots, f_5} \left[ 
    \left( \tParam{G}(f_1, \dots, f_5) \, \tParam{G}(f_1, \dots, f_5){^\dagger} \right) ^{\otimes t}
    \right] \\
    &= \E_{f_1, \dots, f_5} \left[ 
    \frac{1}{2^{nt}}
    \sum_{\bv{i}, \bv{i'} \in (\binset^n)^t}
    \bigotimes_{j=1}^t
    (-1)^{f_1(i_j) - f_1(i'_j)} \,
    V^{f_2, f_3, i_j}_{B_2} \,
    U^{f_5(i_j)}_B
    \rho^{f_4, f_5, i_j, i'_j}_{RABC}
    \left( U^{f_5(i'_j)}_B \right)^{\dagger}
    \left( V^{f_2, f_3, i'_j}_{B_2} \right)^{\dagger}
    \right] 
\end{aligned}
\end{equation}
where we denote $V^{f_2, f_3, i_j}_{B_2} = X^{f_2(i_j)}_{B_2} Z^{f_3(i_j)}_{B_2}$,
\begin{equation*}
    \rho^{f_4, f_5, i_j, i'_j}_{RAB}
    =
    \ketbra{f_5(i_j)}{f_5(i'_j)}_R \otimes \ketbra{\varphi_{f_4(i_j)}}{\varphi_{f_4(i'_j)}}_A \otimes  \ketbra{g_{f_4(i_j)}}{g_{f_4(i'_j)}}_B ~,
\end{equation*}
\begin{equation}
\label{tau_rabc}
    \rho^{f_4, f_5, i_j, i'_j}_{RABC}
    =
    \rho^{f_4, f_5, i_j, i'_j}_{RAB} \otimes  \ketbra{i_j}{i'_j}_C
    ~,
\end{equation}
and writing $\E_{f_i}$ throughout the proof is a shorthand for $\E_{f_i \gets \cF_i}$. We further denote for any pair of $t$-tuples $\bv{i}, \bv{i'}$,
\begin{equation}
\label{tau_rab}
    \tau^{\bv{i}, \bv{i'}}_{RAB} = 
    \E_{f_2, \dots, f_5} \left[
    \bigotimes_{j=1}^t
    V^{f_2, f_3, i_j}_{B_2} \
    U^{f_5(i_j)}_B \
    \rho^{f_4, f_5, i_j, i'_j}_{RAB} \
    \left( U^{f_5(i'_j)}_B \right)^{\dagger}
    \left( V^{f_2, f_3, i'_j}_{B_2} \right)^{\dagger}
    \right]
    ~.
\end{equation}

We define $\hbd_1$ as the projection of $\hbd_0$ onto the distinct subspace of register $C$. This results in the following sub-normalized state.
\begin{equation}
    \hbd_1 \coloneqq \distprojreg{C} \ \hbd_0 \ \distprojreg{C}
\end{equation}

Next, we show that the two hybrids are indistinguishable.
\begin{claim}
\label{tprs_close_to_proj}
    $\trnormdis{\hbd_0}{\hbd_1} = \trnormdis{\hbd_0}{\distprojreg{C} \ \hbd_0 \ \distprojreg{C}} \leq \frac12 \cdot \frac{t^2}{2^n} = \negl(\secp)$.
\end{claim}
\begin{proof}
	We can write
	\begin{flalign*}
        \hbd_0 = 
    	\frac{1}{2^{nt}} \sum_{\bv{i}, \bv{i'}} \E_{f_1} \left[ (-1)^{\sum_{j=1}^t f_1(i_j) - f_1(i'_j)} \tau^{\bv{i}, \bv{i'}}_{RAB} \otimes \ketbra{\bv{i}}{\bv{i'}}_C \right]
        ~.
	\end{flalign*}
	
    Let $\collprojreg{C} = I - \distprojreg{C}$, we start by showing that $\collprojreg{C} \hbd_0\distprojreg{C} = 0$. This follows because
    \begin{align*}
    	\collprojreg{C} \hbd_0\distprojreg{C} &= \frac{1}{2^{nt}} \sum_{\bv{i}, \bv{i'}} \tau^{\bv{i}, \bv{i'}}_{RAB} \otimes \E_{f_1} \left[ (-1)^{\sum_{j=1}^t f_1(i_j) - f_1(i'_j)}\right] \collprojreg{C} \ketbra{\bv{i}}{\bv{i'}}_C \distprojreg{C}
        ~.
    \end{align*}
    Now let us focus on $\E_{f_1} \left[ (-1)^{\sum_{j=1}^t f_1(i_j) - f_1(i'_j)}\right] \collprojreg{C} \ketbra{\bv{i}}{\bv{i'}}_C \distprojreg{C}$. Whenever $\collprojreg{C} \ketbra{\bv{i}}{\bv{i'}}_C \distprojreg{C} \neq 0$ it means that $\bv{i}$ is not distinct and $\bv{i'}$ is distinct. In other words, $\bv{i'}$ has $t$ distinct elements, but $\bv{i}$ has a collision, which means that at least one element in $\bv{i'}$ does not appear in $\bv{i}$ at all. Therefore in this case, since $f_1$ is $2t$-wise independent and $\bv{i} \cup \bv{i'}$ contains at most $2t$ values, then $\E_{f_1} \left[ (-1)^{\sum_{j=1}^t f_1(i_j) - f_1(i'_j)}\right]=0$ and the claim follows. From this property we have

   \begin{flalign*}
	\trnormdis{\hbd_0} {\distprojreg{C} \ \hbd_0 \ \distprojreg{C}}
	&= \trnorm{\collprojreg{C} \ \hbd_0 \ \collprojreg{C}} \\
	&= \tr \left[ \collprojreg{C} \ \hbd_0  \right] \\
	&= \tr \left(
	\frac{1}{2^{nt}} \sum_{\bv{i},\bv{i'} \notin \dist} \E_{f_1}
	\left[ 
	(-1)^{\sum_{j=1}^t f_1(i_j) - f_1(i'_j)}
	\rho_{RAB}^{\bv{i}, \bv{i'}} \otimes \ketbra{\bv{i}}{\bv{i'}}_C
	\right] 
	\right) \\
	&= \frac{1}{2^{nt}} \sum_{\bv{i},\bv{i'} \notin \dist} \E_{f_1}
	\left( 
	(-1)^{\sum_{j=1}^t f_1(i_j) - f_1(i'_j)}
	\right) 
	\tr \left[
	\rho_{RAB}^{\bv{i}, \bv{i'}} 
	\right]
	\tr [\ketbra{\bv{i}}{\bv{i'}}_C] \\
	&\leq \frac{1}{2^{nt}} \sum_{\bv{i} \notin \dist}
	\tr \left[ 
	\rho_{RAB}^{\bv{i}, \bv{i}}
	\right] \\
	&= \frac{\abs{\set{\bv{i} \mid \bv{i} \notin \dist}} }{2^{nt}}
    \end{flalign*}
    which is simply the probability of obtaining a collision when sampling $t$ elements independently from a $2^n$ size universe, which is at most $\binom{t}{2} \cdot \frac{1}{2^n} \leq \frac12 \cdot \frac{t^2}{2^n}$.
\end{proof}

We now give an explicit form for $\hbd_1$.
\begin{claim}
\label{hbd_1_explicit} It holds that
\[
\hbd_1 = \frac{1}{2^{nt}} \sum_{\bv{i} \in \dist, \pi \in S_t} \left( \tau^{\bv{i}, \bv{i}}_{RAB} \otimes \ketbra{\bv{i}}{\bv{i}}_C \right) P_{\pi} ~.\]
\end{claim}
\begin{proof}
    We use the same notation as in Equation~\ref{tau_rab} and write
    \begin{flalign*}
        \hbd_1 &= 
        \distprojreg{C} \left(
        \frac{1}{2^{nt}} \sum_{\bv{i}, \bv{i'} \in (\binset^n)^t} \E_{f_1} \left[ (-1)^{\sum_{j=1}^t f_1(i_j) - f_1(i'_j)} \tau^{\bv{i}, \bv{i'}}_{RAB} \otimes \ketbra{\bv{i}}{\bv{i'}}_C \right]
        \right) \distprojreg{C}\\
    	&= \frac{1}{2^{nt}} \sum_{\bv{i}, \bv{i'} \in \dist} \E_{f_1} \left[ (-1)^{\sum_{j=1}^t f_1(i_j) - f_1(i'_j)} \right] \tau^{\bv{i}, \bv{i'}}_{RAB} \otimes \ketbra{\bv{i}}{\bv{i'}}_C
    \end{flalign*}
    For any $\bv{i}, \bv{i'} \in \dist$, if $\bv{i}$ has an element that does not appear in $\bv{i'}$, or vice versa, then $\E_{f_1} \left[ (-1)^{\sum_{j=1}^t f_1(i_j) - f_1(i'_j)} \right] = 0$. The only terms that don't vanish are those where $\bv{i'}$ is a permutation of $\bv{i}$, and in this case $\E_{f_1} \left[ (-1)^{\sum_{j=1}^t f_1(i_j) - f_1(i'_j)} \right] = 1$. Therefore,
    \begin{flalign*}
        \hbd_1 &= 
    	\frac{1}{2^{nt}} \sum_{\bv{i} \in \dist, \pi \in S_t} \tau^{\bv{i}, \bv{\pi(i)}}_{RAB} \otimes \ketbra{\bv{i}}{\bv{\pi(i)}}_C \\
        &= \frac{1}{2^{nt}} \sum_{\bv{i} \in \dist, \pi \in S_t} \left( \tau^{\bv{i}, \bv{i}}_{RAB} \otimes \ketbra{\bv{i}}{\bv{i}}_C \right) P_{\pi} 
    \end{flalign*}
\end{proof}

The next step is showing that after applying the extractor and a quantum one-time pad in registers $B$ of $\hbd_1$, the resulting sub-normalized state is almost maximally mixed in registers $B$ and $R$.

We formally define 
\begin{equation}
\label{def_hbd_2}
    \hbd_2 \coloneqq 
    \frac{1}{2^{nt}} \sum_{\bv{i} \in \dist, \pi \in S_t} 
    \E_{f_4} \left[ 
    \bigotimes_{j=1}^t
    \frac{I_{R}}{|R|} \otimes \frac{I_{B}}{|B|} \otimes \ketbra{\varphi_{f_4(i_j)}}{\varphi_{f_4(i_j)}}_A \otimes \ketbra{i_j}{i_j}_C 
    \right] P_{\pi}
    ~,
\end{equation}

and show that
\begin{claim}
\label{applying_ext_otp}
    $\trnormdis{\hbd_1}{\hbd_2} \leq 2 t \epsExt = \negl(\secp)$.
\end{claim}

\begin{proof}
    Combining Equation~\ref{tau_rab} and Claim~\ref{hbd_1_explicit} gives
    \begin{flalign*}
        \hbd_1 &= 
        \frac{1}{2^{nt}}  \sum_{\bv{i} \in \dist, \pi \in S_t} 
        \E_{f_4} 
        \E_{f_2, f_3, f_5} \left[
        \bigotimes_{j=1}^t
        V^{f_2, f_3, i_j}_{B_2} \
        U^{f_5(i_j)}_B \
        \rho^{f_4, f_5, i_j, i_j}_{RABC} \
        \left( U^{f_5(i_j)}_B \right)^{\dagger}
        \left( V^{f_2, f_3, i_j}_{B_2} \right)^{\dagger}
        \right] P_{\pi}
        ~.
    \end{flalign*}

    Fixing a $t$-tuple $\bv{i}$ and a function $f_4$, we denote the quantum (normalized) states
    $$\sigma_1^{\bv{i}, f_4}
    = \E_{f_2, f_3, f_5} \left[
    \bigotimes_{j=1}^t
    V^{f_2, f_3, i_j}_{B_2} \
    U^{f_5(i_j)}_B \
    \rho^{f_4, f_5, i_j, i_j}_{RABC} \
    \left( U^{f_5(i_j)}_B \right)^{\dagger}
    \left( V^{f_2, f_3, i_j}_{B_2} \right)^{\dagger}
    \right]
    ~,$$
    and 
    $$\sigma_2^{\bv{i}, f_4}
    = \bigotimes_{j=1}^t
    \frac{I_{R}}{|R|} \otimes \frac{I_{B}}{|B|} \otimes \ketbra{\varphi_{f_4(i_j)}}{\varphi_{f_4(i_j)}}_A \otimes \ketbra{i_j}{i_j}_C 
    ~.$$
     Observe that $\hbd_1 = \frac{1}{2^{nt}}  \sum_{\bv{i} \in \dist, \pi \in S_t} \E_{f_4} \left[ \sigma_1^{\bv{i}, f_4} \right] P_{\pi}$ and $\hbd_2 = \frac{1}{2^{nt}}  \sum_{\bv{i} \in \dist, \pi \in S_t} \E_{f_4} \left[ \sigma_2^{\bv{i}, f_4} \right] P_{\pi}$.
     
     We will show that $\td{\sigma_1^{\bv{i}, f_4}}{\sigma_2^{\bv{i}, f_4}} \leq t \epsExt$ and then conclude from Proposition~\ref{prop_norm_right_mult} and Proposition~\ref{prop_distance_expect} that $\trnormdis{\hbd_1}{\hbd_2} \leq 2 t \epsExt$.

     For a single index $i_j$ note that the functions $f_2, f_3, f_5$ act as random functions, and denote
     $$\rho_{ABC}^{f_4, i_j} = \ketbra{\varphi_{f_4(i_j)}}{\varphi_{f_4(i_j)}}_{A}
    \otimes \ketbra{g_{f_4(i_j)}}{g_{f_4(i_j)}}_{B}
    \otimes \ketbra{i_j}{i_j}_{C} ,$$
    $$\rho_{AC}^{f_4, i_j} = \ketbra{\varphi_{f_4(i_j)}}{\varphi_{f_4(i_j)}}_{A} \otimes \ketbra{i_j}{i_j}_{C} ~.$$
    Observe that each $\rho_{ABC}^{f_4, i_j}$ is a pure state with $H_{\infty}(B | AC) = 0$, so for any $\delta$, $H_{\infty}^{\delta}(B | AC) \geq 0$. Applying Theorem~\ref{p_extractor_thm} with $\nExt = \garbqubits$, $\kExt = 0$ and $\epsExt$, we can extract $\ell$ qubits from $B$ with the guarantee that
    \begin{equation*}
        \td{
        \E_{f_5} \left[
        \ketbra{f_5(i_j)}{f_5(i_j)}_{R} \otimes
        \tr_{B_2} \left( U^{f_5(i_j)}_B \rho_{ABC}^{f_4, i_j} \left(U^{f_5(i_j)}_B\right)^{\dagger} \right)
        \right]
        }{
        \frac{I_{R}}{|R|} \otimes \frac{I_{B_1}}{|B_1|} \otimes \rho_{AC}^{f_4, i_j}
        } \leq \epsExt
        ~.
    \end{equation*}
    Tensoring with the maximally mixed state on register $B_2$ gives
    \begin{equation*}
        \td{
        \underbrace{
        \E_{f_5} \left[
        \ketbra{f_5(i_j)}{f_5(i_j)}_{R} \otimes
        \tr_{B_2} \left( U^{f_5(i_j)}_B \rho_{ABC}^{f_4, i_j} \left(U^{f_5(i_j)}_B\right)^{\dagger} \right)
        \otimes \frac{I_{B_2}}{|B_2|}
        \right]
        }_{\coloneqq \sigma_1^{i_j, f_4}}
        }{
        \underbrace{
        \frac{I_{R}}{|R|} \otimes \frac{I_{B_1}}{|B_1|} \otimes \rho_{AC}^{f_4, i_j}
        \otimes \frac{I_{B_2}}{|B_2|}
        }_{\coloneqq \sigma_2^{i_j, f_4}}
        } \leq \epsExt
    \end{equation*}

    We start by analyzing $\sigma_1^{i_j, f_4}$.
    Observe that for each $r$, quantum OTP gives (Proposition~\ref{prop_otp})
    \begin{equation*}
    \begin{split}
        & \E_{f_2, f_3} \left[ \ketbra{r}{r}_R \otimes
        V^{f_2, f_3, i_j}_{B_2} \
        U^r_B \
        \rho_{ABC}^{f_4, i_j} 
        \left(U^r_B\right)^{\dagger}
        \left( V^{f_2, f_3, i_j}_{B_2} \right)^{\dagger} \right]
        \\
        &= \ketbra{r}{r}_R \otimes \tr_{B_2} \left( 
        U^r_B
        \rho_{ABC}^{f_4, i_j}
        \left(U^r_B\right)^{\dagger}
        \right)
        \otimes \frac{I_{B_2}}{|B_2|}
    \end{split}
    \end{equation*}
    So
    \begin{equation*}
    \begin{split}
         \sigma_1^{i_j, f_4}
        &= \E_{f_5} \left[
        \ketbra{f_5(i_j)}{f_5(i_j)}_R \otimes \tr_{B_2} \left( 
        U^{f_5(i_j)}_B
        \rho_{ABC}^{f_4, i_j}
        \left(U^{f_5(i_j)}_B\right)^{\dagger}
        \right)
        \otimes \frac{I_{B_2}}{|B_2|}
        \right] \\
        &= \E_{f_2, f_3, f_5} \left[ 
        \ketbra{f_5(i_j)}{f_5(i_j)}_R \otimes
        V^{f_2, f_3, i_j}_{B_2} \
        U^{f_5(i_j)}_B \
        \rho_{ABC}^{f_4, i_j} \,
        \left(U^{f_5(i_j)}_B\right)^{\dagger} \,
        {V^{f_2, f_3, i_j}_{B_2}}^{\dagger}
        \right]
    \end{split}
    \end{equation*}
    Now we analyze $\sigma_2^{i_j, f_4}$.
    \begin{equation*}
    \begin{split}
        \sigma_2^{i_j, f_4} 
        &= \frac{I_{R}}{|R|}
        \otimes \rho_{AC}^{f_4, i_j} \otimes \frac{I_{B_1}}{|B_1|}
        \otimes \frac{I_{B_2}}{|B_2|} \\
        &= \frac{I_{R}}{|R|}
        \otimes \frac{I_{B}}{|B|}
        \otimes \rho_{AC}^{f_4, i_j}
    \end{split}
    \end{equation*}
    
    Now consider all the elements of the $t$-tuple $\bv{i} = (i_1, \dots, i_t)$. From Proposition~\ref{prop_distance_pairs_tensor}, 
    $$\td{ \bigotimes_{j=1}^t \sigma_1^{i_j, f_4}
    }{
    \underbrace{
    \bigotimes_{j=1}^t \sigma_2^{i_j, f_4} 
    }_{\sigma_2^{\bv{i}, f_4}}
    } \leq t \epsExt
    ~.$$
    Since $\bv{i}$ contains $t$ distinct elements, and $f_2, f_3, f_5$ are $t$-wise independent functions, we can apply Proposition~\ref{prop_indp_expect} on $\bigotimes_{j=1}^t \sigma_1^{i_j, f_4}$, which will give us exactly $\sigma_1^{\bv{i}, f_4}$.
\end{proof}

We define $\hbd_3$ by replacing each of the $1$-PRS states in register $A$ of $\hbd_2$ with a maximally mixed state. Formally,
\begin{equation}
    \hbd_3 \coloneqq
    \frac{1}{2^{nt}} \sum_{\bv{i} \in \dist, \pi \in S_t} 
    \left(
    \bigotimes_{j=1}^t \frac{I_{R}}{|R|}
    \otimes \frac{I_{B}}{|B|}
    \otimes \frac{I_{A}}{|A|}
    \otimes \ketbra{i_j}{i_j}_{C}
    \right) P_{\pi}
    ~.
\end{equation}

We use the security of the $1$-PRS to show that
\begin{claim} For any distinguishing adversary $\cA$, 
    $\abs{
    \Pr \left[ \cA(\hbd_2) = 1 \right]
    -
    \Pr \left[ \cA(\hbd_3) = 1 \right]
    }
    \leq \negl(\secp)
    $.
\end{claim}
\begin{proof}
    As defined in Equation~\ref{def_hbd_2},
    $$\hbd_2
    = \frac{1}{2^{nt}} \sum_{\bv{i} \in \dist, \pi \in S_t} \E_{f_4}
    \left[
    \bigotimes_{j=1}^t \frac{I_{R}}{|R|}
    \otimes \frac{I_{B}}{|B|}
    \otimes \ketbra{\varphi_{f_4(i_j)}}{\varphi_{f_4(i_j)}}_{A}
    \otimes \ketbra{i_j}{i_j}_{C}
    \right] P_{\pi}
    ~.$$
    Since $\bv{i} = (i_1, \dots, i_t)$ are distinct and $f_4$ is $t$-wise independent, this equals
    \begin{flalign*}
        = \frac{1}{2^{nt}} \sum_{\bv{i} \in \dist, \pi \in S_t} \E_{\bv{s} \gets (\binset^\keylen)^t}
        \left[
        \bigotimes_{j=1}^t \frac{I_{R}}{|R|}
        \otimes \frac{I_{B}}{|B|}
        \otimes \ketbra{\varphi_{s_j}}{\varphi_{s_j}}_{A}
        \otimes \ketbra{i_j}{i_j}_{C}
        \right] P_{\pi}
        ~.
    \end{flalign*}
    We would like to use the $1$-PRS security and ``replace'' each instance of $\ket{\varphi_{s_j}}$ with a Haar-random state. In other words, we want a reduction from a QPT distinguishing adversary that receives this state to a $1$-PRS QPT distinguisher. However, there are two issues we need to overcome. The first is that we are using states that are sub-normalized, and quantum reductions use normalized states. We overcome this by considering the normalized version of this state, as well as the normalized version of the state that is the result of replacing each $\ket{\varphi_{s_j}}$ by a random computational basis state, which is precisely $\hbd_3$.
    \begin{flalign*}
        &\quad
        \frac{1}{2^{nt}} \sum_{\bv{i} \in \dist, \pi \in S_t} \E_{\bv{y} \gets (\binset^\outqubits)^t}
        \left[
        \bigotimes_{j=1}^t \frac{I_{R}}{|R|}
        \otimes \frac{I_{B}}{|B|}
        \otimes \ketbra{y_j}{y_j}_{A}
        \otimes \ketbra{i_j}{i_j}_{C}
        \right] P_{\pi} \\
        &= \frac{1}{2^{nt}} \sum_{\bv{i} \in \dist, \pi \in S_t}
        \left(
        \bigotimes_{j=1}^t \frac{I_{R}}{|R|}
        \otimes \frac{I_{B}}{|B|}
        \otimes \frac{I_{A}}{|A|}
        \otimes \ketbra{i_j}{i_j}_{C}
        \right) P_{\pi} \\
        &= \hbd_3
    \end{flalign*}
    Notice that we used random computational basis states instead of Haar-random states, since for single-copy security they share the same density matrix. Furthermore, since $\tr[\hbd_2] = \tr[\hbd_3]$, then for any QPT distinguishing adversary $\cA$, its distinguishing advantage of the normalized states is at least that of the non-normalized states, namely
    \begin{flalign*}
        &\quad 
        \abs{\Pr[\cA(\hbd_2) = 1] - \Pr[\cA(\hbd_3) = 1]} \\
        &= \abs{\Pr \left[ \cA \left( \frac{\hbd_2}{\tr[\hbd_2]} \right) = 1 \right] \cdot \tr[\hbd_2]
        - \Pr \left[ \cA \left( \frac{\hbd_3}{\tr[\hbd_3]} \right) = 1 \right] \cdot \tr[\hbd_3]} \\
        &\leq \abs{\Pr \left[ \cA \left( \frac{\hbd_2}{\tr[\hbd_2]} \right) = 1 \right]
        - \Pr \left[ \cA \left( \frac{\hbd_3}{\tr[\hbd_3]} \right) = 1 \right]} \\
        &= \abs{\Pr \left[ \cA \left( \rho_2 \right) = 1 \right]
        - \Pr \left[ \cA \left( \rho_3 \right) = 1 \right]}
    \end{flalign*}
    where we define $\rho_2 \coloneqq \frac{\hbd_2}{\tr[\hbd_2]}$ and $\rho_3 \coloneqq \frac{\hbd_3}{\tr[\hbd_3]}$.
    
    The second issue is how to efficiently simulate these normalized states given independent instances from a $1$-PRS distinguisher. For this we will use the symmetrization simulator defined in Subsection~\ref{symmetrization_simulator}. For generating $\rho_2$ sample independently $\bv{r} \gets (\binset^\designseedlen)^t, \bv{b} \gets (\binset^\garbqubits)^t, \bv{s} \gets (\binset^\keylen)^t$ and
    run the simulator on the state $\bigotimes_{j=1}^t \ket{r_j}_R \ket{b_j}_B \ket{\varphi_{s_j}}_A$. By Claim~\ref{simulator_output} the output state will be 
    \begin{flalign*}
        &\quad 
        \E_{\bv{r}, \bv{b}, \bv{s}} \left[ 
        \tsim{ 
        \bigotimes_{j=1}^t \ket{r_j}_R \ket{b_j}_B \ket{\varphi_{s_j}}_A} 
        \right] \\
        &= \E_{\bv{r}, \bv{b}, \bv{s}} \left[
        \frac{1}{\abs{\dist}} \sum_{\bv{i} \in \dist, \pi \in S_t}
        \left( \bigotimes_{j=1}^t \ketbra{r_j}{r_j}_R \otimes \ketbra{b_j}{b_j}_B \otimes \ketbra{\varphi_{s_j}}{\varphi_{s_j}}_A \otimes \ketbra{i_j}{i_j}_C \right) P_{\pi} 
        \right] \\
        &=  \frac{1}{\abs{\dist}} \sum_{\bv{i} \in \dist, \pi \in S_t}
        \E_{\bv{s}} \left[
        \bigotimes_{j=1}^t \frac{I_R}{\abs{R}}_R \otimes \frac{I_B}{\abs{B}}_B \otimes \ketbra{\varphi_{s_j}}{\varphi_{s_j}}_A \otimes \ketbra{i_j}{i_j}_C \right] P_{\pi} \\
        &= \frac{2^{nt}}{\abs{\dist}} \hbd_2 = \rho_2
    \end{flalign*}
    For generating $\rho_3$, sample independently $\bv{r} \gets (\binset^\designseedlen)^t, \bv{b} \gets (\binset^\garbqubits)^t, \bv{y} \gets (\binset^\outqubits)^t$ and run the simulator on the state $\bigotimes_{j=1}^t \ket{r_j}_R \ket{b_j}_B \ket{y_j}_A$. A similar calculation shows that in this case the output state will be $\rho_3$. 
    
    Hence now we can reduce to the security of $1$-PRS. If a QPT distinguishing adversary $\cA$ can distinguish $\hbd_2$ from $\hbd_3$ with non-negligible advantage then it can distinguish $\rho_2$ from $\rho_3$ with at least the same advantage. Since $\rho_2$ can be simulated given $t$ independent copies of single-copy pseudorandom states and $\rho_3$ can be simulated given $t$ independent states from the computational basis, there exists a QPT distinguishing adversary that distinguishes $t$ independent copies of the $1$-PRS from $t$ independent copies of random states with non-negligible advantage, simply by running $\cA$ on the output of the simulator. Applying a standard hybrid argument as in Claim~\ref{prop_1prs_amplification} to this adversary, we can replace the $t$ coordinates one by one. It follows that the distinguishing advantage in at least one coordinate must be at least a $1/t$ fraction of the total advantage (and thus still non-negligible), yielding a $1$-PRS distinguisher that contradicts $1$-PRS security. Therefore it must be that $\abs{\Pr[\cA(\hbd_2) = 1] - \Pr[\cA(\hbd_3) = 1]} \leq \negl(\secp)$.
\end{proof}

Finally, recall that the density matrix of a $t$-copy Haar random state equals the normalized projector onto the symmetric subspace. Hence we formally define our last hybrid as
\begin{equation}
    \hbd_4 \coloneqq \E_{\ket{\psi} \gets \mu_{\tParam{\outqubits}}} \left[ \ketbra{\psi}{\psi}^{\otimes t} \right] = \symrho
    ~.
\end{equation}

We show that
\begin{claim}
\label{close_to_haar}
    $\trnormdis{\hbd_3}{\hbd_4} \leq 3 \frac{t^2}{2^n} = \negl(\secp)$.
\end{claim}

\begin{proof}
    We develop the expression of $\hbd_3$. Denote $w \coloneqq \designseedlen + \garbqubits + \outqubits$.
    \begin{flalign*} 
        \hbd_3 
        &= \frac{1}{2^{nt}} \sum_{\bv{i} \in \dist, \pi \in S_t}
        \left(
        \bigotimes_{j=1}^t \frac{I_{R}}{|R|}
        \otimes \frac{I_{B}}{|B|}
        \otimes \frac{I_{A}}{|A|}
        \otimes \ketbra{i_j}{i_j}_{C}
        \right) P_{\pi} \\
        &= \frac{1}{2^{nt}} \sum_{\bv{i} \in \dist, \pi \in S_t} \E_{\bv{x} \gets (\binset^w)^t}
        \left[
        \bigotimes_{j=1}^t \ketbra{x_j}{x_j}_{RBA}
        \otimes \ketbra{i_j}{i_j}_{C}
        \right] P_{\pi} \\
        &= \frac{1}{2^{nt} \cdot 2^{wt}}
        \sum_{\bv{i} \in \dist, \bv{x} \in (\binset^w)^t, \pi \in S_t}
        \left[
        \left(
        \bigotimes_{j=1}^t \ketbra{x_j}{x_j}_{RBA}
        \otimes \ketbra{i_j}{i_j}_{C}
        \right) P_{\pi}
        \right] \\
        &=  \frac{t!}{2^{\tParam{\outqubits} t}} \symproj \distprojreg{C}
    \end{flalign*}
    where $\symproj$ is the projector onto the symmetric subspace of the whole system, and $\distprojreg{C}$ is the projector onto the distinct subspace of system $C$.
    
    Denote $d = 2^{\tParam{\outqubits}}$, and recall the normalized projector onto the symmetric subspace $\symrho = \frac{\symproj}{\binom{d + t - 1}{t}}$. Notice that 
    \begin{flalign*}
        \trnormdis{\frac{t!}{d^t} \symproj}{\symrho}
        &= \trnorm{\symrho \left( \frac{t!}{d^t} \binom{d + t - 1}{t} - 1 \right)} \\
        &= \frac{t!}{d^t} \binom{d + t - 1}{t} - 1 \\
        &= \frac{(d + t - 1)!}{d^t (d - 1)!} - 1 \\
        &= \frac{(d + t - 1)(d + t - 2) \dots d}{d^t} - 1 \\
        &\coloneqq \delta_0
    \end{flalign*}
    and since $\distprojreg{C}$ is a projector, 
    $$\trnormdis{\frac{t!}{d^t} \symproj \distprojreg{C}}{\symrho \distprojreg{C}}
    \leq \trnormdis{\frac{t!}{d^t} \symproj}{\symrho}
    = \delta_0 ~.$$
    Claim~\ref{proj_to_dis_subspace_close_to_sym} gives $\trnormdis{\symrho}{\symrho \distprojreg{C}} \leq \frac{t^2}{2^n}$ as long as $\frac{2t}{2^n} \leq 1$, which is true for large enough $\secp$. Hence by triangle inequality we get
    \begin{flalign*}
        \trnormdis{\hbd_3}{\hbd_4}
        = \trnormdis{
        \frac{t!}{d^t} \symproj \distprojreg{C}
        }{
        \symrho
        }
        \leq \delta_0 + \frac{t^2}{2^n}
        ~.
    \end{flalign*}
    We conclude by showing that $\delta_0 \leq 2 \frac{t^2}{2^n}$.
    \begin{flalign*}
    \delta_0 
    &= \frac{(d + t - 1)(d + t - 2) \dots d}{d^t} - 1 \\
    &= \left[ \prod_{j=0}^{t-1} \left( 1 + \underbrace{\frac{j}{d}}_{\epsilon_j} \right) \right] - 1 
    \end{flalign*}
    By Proposition~\ref{mult_1_plus_eps_i}, if $\epsilon \coloneqq \sum_{j=0}^{t-1} \epsilon_j < 1$ then $\delta_0 \leq \epsilon + \epsilon^2 \leq 2 \epsilon$. Asymptotically, $\frac{t^2}{d} = \negl(\secp)$, so for large enough $\secp$ we have $\sum_{j=0}^{t-1} \epsilon_j = \sum_{j=0}^{t-1} \frac{j}{d} \leq \frac{t^2}{d} < 1$. Hence $\delta_0 \leq 2\frac{t^2}{d} \leq 2\frac{t^2}{2^n} = \negl(\secp)$.
\end{proof}

\paragraph*{Acknowledgments.}
The authors are supported by the Horizon Europe Research and Innovation Program via ERC Project ACQUA (Grant 101087742).

\printbibliography

\end{document}